\documentclass[12pt]{article}
\usepackage[utf8]{inputenc}
\usepackage{amsmath}
\usepackage{amsthm}
\usepackage{amssymb} 
\usepackage{multirow}
\usepackage{tikz}
\tikzstyle{bag} = [align=center]
\usetikzlibrary{arrows}
\usetikzlibrary{trees}
\usepackage{multirow,array}
\usepackage{hyperref}

\usepackage{geometry}
\newtheorem{definition}{Definition}[section]
 
\newtheorem{theorem}[definition]{Theorem}
\newtheorem{corollary}{Corollary}[definition]
\newtheorem{lemma}[definition]{Lemma}

\title{Asymmetric Equilibria in Symmetric Multiplayer Prisoners' Dilemma Supergames} 
\author{Davidson Cheng\footnote{Colorado College, Colorado Springs, CO. Email:  d\_cheng@coloradocollege.edu.}}

\date{April 18, 2022}

\begin{document}

\maketitle

\begin{abstract}
    We propose a finite automaton-style solution concept for supergames. In our model, we define an equilibrium to be a cycle of state switches and a supergame to be an infinite walk on states of a finite stage game. We show that if the stage game is locally non-cooperative, and the utility function is monotonously decreasing as the number of defective agents increases, the symmetric multiagent prisoners' dilemma supergame must contain one symmetric equilibrium and can contain asymmetric equilibria. 
\end{abstract}

\section{Introduction}

The prisoners' dilemma is one of the most-studied topics in game theory as it abstracts many real-life game scenarios. At the same time, the prisoner's dilemma delivers the pessimistic message that cooperation is difficult to reach with selfish, rational agents. Since cooperation is commonplace in reality, this indicates that the original prisoner's dilemma might not be a sufficient model for cooperation. Interestingly, when the prisoner's dilemma is repeated, some unexpected behaviors begin to occur. In 1980, Robert Axelrod created a ``prisoner's dilemma tournament" where he invited game theory scholars at the time to participate \cite{Axelrod1981}. The tournament consists of several hundred rounds, where the invited players created proxies with computer programs and are allowed to draw on past results while making decisions. Surprisingly, the winner was a simple strategy named ``tit for tat" that played cooperation until the other player defects. What was equally surprising was that cooperation between the players was largely common during this tournament. Axelrod's experiment revealed that repetition of a single-stage non-cooperative game may promote cooperation. Repeated games have since then been regarded as a more suitable medium for modeling evolution and cooperation. Theoretical results from studying repeated games have been applied to explain phenomenons of cooperation in evolution and human interactions \cite{Fudenberg_1990}. 

A symmetric game is one where the payoff function is indiscriminate toward the players. Similarly, a symmetric equilibrium is one such that not every agent adopts the same strategy. It was shown recently that despite every finite symmetric game must contain at least one symmetric nash equilibrium \cite{Nash1951}, its equilibria are not all required to be symmetric \cite{FEY2012424} \cite{Xefteris_Zero_Sum}. We then ask the question of whether repeated symmetric games have only symmetric equilibrium. The answer to this question can help us understand many real world phenomenons where heterogeneous actions are observed under a non-discriminatory policy. Additionally, we might gain insight into whether asymmetric strategies are purely results of individual agent differences, or whether they are produced partially by the policy as well. For example, when students are assigned a group project, it's often observed that a few students in a group end up doing most of the work while the rest become freeriders. Obviously, students who care more about their grade will put in more effort, but the question we are interested in is whether the policy (grading rubric in this case) contributes to the differences in student efforts.

In this paper, we introduce a solution concept for supergames that is similar to the finite state automaton model. We then show that under reasonable assumptions, the set of equilibria for a multi-agent prisoners' dilemma supergame must contain one symmetric equilibrium and can contain asymmetric equilibria. In section 2, we provide a brief introduction to repeated games and define our model. Then, we describe our solution concept and prove our main theorem by characterizing the equilibria of two examples.

\section{Model of Repeated Games}

There are many scenarios where agents need to take past and future interactions between them and other agents into consideration when making a decision. Repeated games are often used to capture this notion. It is well-known that equilibrium can change when games are repeated \cite{Fudenberg_1990}, but we still don't understand the dynamics of these changes. 
 
 Existing studies on repeated games can be divided into two categories: finitely repeated and infinitely repeated. As the names suggest, the finitely repeated games are defined on a finite time horizon while the infinitely repeated games are played over an infinite horizon. At first glance, finitely repeated games may be better justified as a realistic model, since infinite life span is impractical. However, finitely repeated games often do not reach cooperation when expected to \cite{benoit1984finitely}. On the other hand, infinitely repeated games often contain a larger set of equilibrium than its finitely repeated counterparts. In fact, the folk theorem states that every individually rational and feasible payoff vector is a Nash equilibrium in an infinitely repeated game without discounting. It has been noticed that infinitely repeated games produce different dynamics compared to stage games and finitely repeated games \cite{RUBINSTEIN19791} \cite{benoit1984finitely}, and may be a more suitable approach to study the arise of cooperation in the long run. 

We will define a class of symmetric repeated games by first specifying a class of symmetric multiplayer stage games. The stage game can take on any arbitrary number of players, and we restrict players' actions to ``cooperate" and ``defect." This class of games is also known as multiplayer prisoners' dilemma. Let 

\[G=<\{0,1\},S_n, u>\] 

be a symmetric stage games involving $n \ge 2$ players where each player has the pure strategy profile of 1 (cooperate) and 0 (defect). The set $S_n$ contains all possible states of the game whereas a state is defined as the collection of all agents' chosen pure strategies. The payoff for each player is determined by $u: \{0,1\} \times S_n \mapsto \mathbb{R}$. Given the symmetric property of $G$, we can specify any state with an integer that counts the number of the ``defective" agents. We will notate a state $s \in S_n$ with $[j]$ for some $j \in \{0,1,...,n\}$ whereas $j$ is the number of defective agents in the state $s$. Let $m_i^k \in \{0,1\}$ and $s^k \in S_n$ denote the $k$-th round strategy and state, respectively. When the game $G$ is repeated, the payoff of player $i$ at some round $k$ can be computed by invoking $u(m_i^k,s^k)$. The cumulative payoff of the $i$-th player until the $k$-th round can be given as


\[\sum_{h = 0}^{k} u(m_i^h,s^h).\]





We will focus specifically on infinitely repeated games without discounting, this class of games is also known as ``supergames," denoted with $G^*$. Because of the infinite horizon in supergames, cumulative utility is no longer a valid approach for measuring payoff of a strategy. Instead, the ``limit of the means" is commonly adopted for payoff in supergames \cite{aumann1994long}, that is 


\[\text{payoff}_i = \lim_{k \rightarrow \infty} \frac{1}{k}\sum_{h = 0}^{k} u(m_i^h,s^h).\]





In stage games, to check whether a strategy vector is a nash equilibrium, it suffices to check if at least one agent can by themself improve their payoff. This method is not sufficient for supergames. In supergames, finding equilibria involves deriving the equilibrium that a state will eventually reach. Let's say the state $s_e$ is an equilibrium, and the game takes $k$ rounds to transition from its initial state to $s_e$, then the limit of the means payoff over the infinite horizon can be given as

\[\lim_{r \rightarrow \infty} \frac{1}{r} \bigg ( \sum_{h=0}^{k}u(m^h,s^h) + \sum_{i = k+1}^{r} u(0,s_e) \bigg ) = u(0,s_e) \text{ for cooperative players at $s_e$},\]

\[\lim_{r \rightarrow \infty} \frac{1}{r} \bigg ( \sum_{h=0}^{k}u(m^h,s^h) + \sum_{i = k+1}^{r} u(1,s_e) \bigg ) = u(1,s_e) \text{ for defective players at $s_e$}.\]

Notice that this expression disregards the finite time-span prior to reaching the equilibrium, which is called the overtaking period. Rubinstein \cite{RUBINSTEIN19791} has argued that despite two sequences of game dynamics may reach the same equilibrium, the overtaking period can allow one to dominate the other. Here, we stick to the limiting means payoff for simplicity sake, as such difference is marginal over infinite horizon.

The above expressions indicate that the objective for rational players in $G^*$ is to let the game reach an equilibrium state that maximizes their payoff, which can be a complicated recursive process. In the next section, we will introduce our solution concept to alleviate this complexity.

\section{Emergence of Asymmetric Equilibrium}

In the previous section, we have identified that an agent in $G^*$ aims to induce the game into an equilibrium state that favors them. When numerous agents like such are interacting together, it may not be straightforward to determine exactly what equilibrium state the game will eventually reach. We now introduce a few concepts to tackle this issue. Our solution concept consider the stage game $G$ to be a finite-state machine, and the supergame $G^*$ as a infinite walk on such machine. The dynamics of $G^*$ can then be described with a finite sequence of state switches that enter an infinite cycle after reaching equilibrium.

\subsection{Leading and Chaining}

\begin{definition}
    For some $s_i,s_j \in S_n$, we use $s_i \rightarrow s_j$ to denote ``$s_i$ leads to $s_j$", which means one and only one agent in $s_i$ is able to by themself improve their limit of the means payoff through inducing a state switch from $s_i$ to $s_j$, or $s_i = s_j$ and $s_i$ is an equilibrium. 
\end{definition}

In other words, if we end up with the state $s_i$ in some round $k$, and $s_i \rightarrow s_j$, we can expect the game to evolve into $s_j$ in round $k+1$. We find the ``one and only one" restriction necessary to develop the rest of our analysis. Since the game is symmetric, if one cooperative/defective agent in a state $s$ is able to improve their overall utility by some strategy switch, every other cooperative/defective agent are able to as well. So unless we introduce a policy (either probabilistic or asymmetrical) to determine how many of these agents we allow to switch, the game will immediately enter a symmetric state, which makes the supergame trivial. One way to interpert the concept of leading is that at every round, we are picking an agent from a uniform probability distribution and ask them if they would like to switch their current strategy.

\begin{corollary}
    If $[b] \rightarrow [c]$, then $b-1 \le c \le b+1$.
\end{corollary}

\begin{definition}
    For $s_i, s_j \in S_n$, $s_i \Rightarrow s_j$ denotes ``$s_i$ is chained to $s_j$", which means  for some $s_i, s_l, s_m, ... s_o, s_j \in S_n$, $s_i \rightarrow s_l$, $s_l \rightarrow s_m$, ..., $s_o \rightarrow s_j$. 
\end{definition}

Consider $G^*$ to be a directed muiltigraph on $S_n$, then $s_i \Rightarrow s_j$ means there exists a path from $s_i$ to $s_j$. In other words, if the game reaches $s_i$, we expect the game to evolve into $s_j$ within a finite number of subsequent rounds. 

\begin{lemma}
    For $[b], [c] \in S_n$, if  $[b] \Rightarrow [c]$ and $c-1 \le b \le c+1$, then $[b] \rightarrow [c]$.
\end{lemma}

\begin{proof}
    If $[b] \Rightarrow [c]$ then it is implied that one agent is able to improve their overall utility by inducing a sequence of state switches and eventually reach $[c]$. Given $c-1 \le b \le c+1$, such transition can be done within a single round, thus $[b] \rightarrow [c]$.
\end{proof}

\begin{lemma}
\label{monotone lemma}
    For $[b], [c] \in S_n$, if  $[b] \Rightarrow [c]$ and WLOG $c > b$, then for any $d$ satisfying $b \le d < c$, $[d] \rightarrow [d + 1]$
\end{lemma}

\begin{proof}
    Suppose for some $d$ satisfies $b \le d < c$, $[d] \not \rightarrow [d + 1]$. Then for any state $[f]$ such that $[b] \Rightarrow [f]$, it is implied that $[f] \not \Rightarrow [c]$, which means $[b] \not \Rightarrow [c]$.
\end{proof}

\begin{corollary}
    WLOG, if $b \le c$ and $[b] \Rightarrow [c]$, then for any $d$ satisfying $b < d \le c$, $[b] \Rightarrow [d]$.
\end{corollary}

Lemma 3.4 and Corollary 3.4.1 indicates that in $G^*$, the sequence of state switches from some state to its chained equilibrium state will contain a monotone chain reaction, like depicted in the following diagram. 
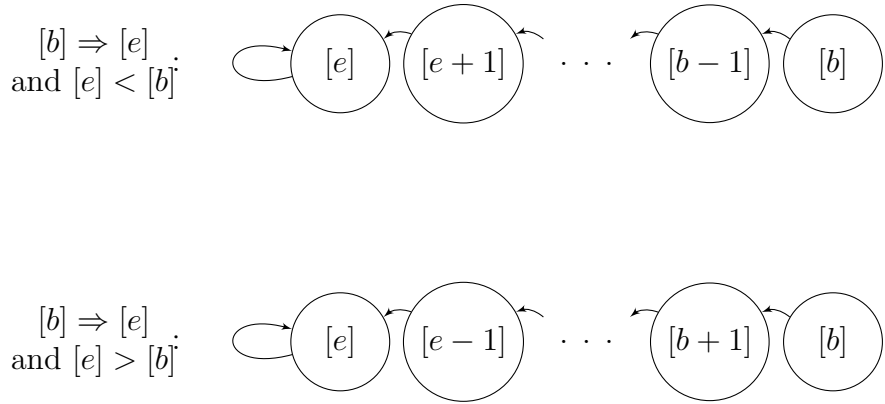
\vspace{0.5cm}
\begin{figure}[h]
\centering
    \begin{tikzpicture}
  [scale=.54,auto=left,every node/.style={draw, circle, thin, minimum size=1.3cm}]
  
  \tikzset{edge/.style = {->,> = latex'}}
  \node[bag] (n1) at (17,0) {$[b]$};
  \node[bag] (n2) at (14,0)  {$[b-1]$};
  \node[bag, draw = none] (n3) at (11,0) {. . .};
  \node[bag] (n4) at (8,0) {$[e+1]$};
  \node[bag] (n5) at (5,0) {$[e]$};
  
  \node[bag, draw = none] (n6) at (-1,0) {$[b] \Rightarrow [e] $ \\ and $[e] < [b] $};
  \node[bag, draw = none] (n7) at (1,0) {:};
  
  
  \draw[edge] (n1)  to[bend right] (n2);
  \draw[edge] (n2)  to[bend right] (n3);
  \draw[edge] (n3)  to[bend right] (n4);
  \draw[edge] (n4)  to[bend right] (n5);
  \draw[edge] (n5)  to[loop left] (n5);

\end{tikzpicture}
\end{figure}

\begin{figure}[h]
\centering
    \begin{tikzpicture}
  [scale=.54,auto=left,every node/.style={draw, circle, thin, minimum size=1.3cm}]
  
  \tikzset{edge/.style = {->,> = latex'}}
  \node[bag] (n1) at (17,0) {$[b]$};
  \node[bag] (n2) at (14,0)  {$[b+1]$};
  \node[bag, draw = none] (n3) at (11,0) {. . .};
  \node[bag] (n4) at (8,0) {$[e-1]$};
  \node[bag] (n5) at (5,0) {$[e]$};
  
  \node[bag, draw = none] (n6) at (-1,0) {$[b] \Rightarrow [e] $ \\ and $[e] > [b] $};
  \node[bag, draw = none] (n7) at (1,0) {:};
  
  
  \draw[edge] (n1)  to[bend right] (n2);
  \draw[edge] (n2)  to[bend right] (n3);
  \draw[edge] (n3)  to[bend right] (n4);
  \draw[edge] (n4)  to[bend right] (n5);
  \draw[edge] (n5)  to[loop left] (n5);

\end{tikzpicture}
\end{figure}

\subsection{Game Properties}

We now introduce two properties that the multiplayer prisoners' dilemma can take on. 


\begin{definition}
    The game $G$ is locally non-cooperative if 
    
    \begin{itemize}
        \item for any state with utility defined for both defective and cooperative agents, defective agents have higher payoff.
        \item for any $b \in \{0,1,2,...,n-1\}$, 
    
    \[u(1,[b+1]) > u(0,[b]).\]
    \end{itemize}   

\end{definition}

In the two symmetric states of $G^*$, namely $[0]$ and $[n]$, only one of the utility is specified (there is no cooperative agent at $[n]$ and no defective agent at $[0]$), therefore we excluded those two states from the above definition. Given a locally non-cooperative supergame, it is implied that myopic agents would always prefer defect. However, as we will show in later sections, this is not the case for farsighted, rational agents. 

\begin{definition}
    The supergame $G^*$ has monotonously decreasing utility function if for any $s$, $s' \in S_n$ such that $s' > s$, the utility function satisfy
    
    \[u(0,s') < u(0,s), \text{   and}\]
    
    \[u(1,s') < u(1,s).\]
\end{definition}

Here, we define defective to be the higher ordered action and states with more defective agents to be higher ordered states. So $[a] > [b]$ if and only if $a > b$. To see the practical significance of monotonous decreasing utility, consider a scenario where the welfare is maximized at the symmetric cooperative state, and the welfare of this game is monotonously decareasing as the number of defective agents increases.

    






\subsection{Transition Properties of $G^*$}

\begin{theorem}
    If $G^*$ has monotonously decreasing utility function and is locally  non-cooperative, then every state leads to exactly one state.
\end{theorem}

\begin{proof}
    There only exist four possible cases for a state $[b]$ to lead to more than one state.
    
    \begin{itemize}
    \item[1.] $[b] \rightarrow [b]$, $[b] \rightarrow [b+1]$, $[b] \rightarrow [b-1]$ 
    \item[2.] $[b] \rightarrow [b]$, $[b] \rightarrow [b+1]$
    \item[3.] $[b] \rightarrow [b]$, $[b] \rightarrow [b-1]$
    \item[4.] $[b] \rightarrow [b+1]$, $[b] \rightarrow [b-1]$
\end{itemize}

    By definition, a state $[b]$ leads to $[b]$ if and only if $[b]$ is an equilibrium, meaning no single player is able to single-handedly improve their equilibrium utility at $[b]$. Notice that $[b] \rightarrow [b+1]$ or $[b] \rightarrow [b-1]$ suggests the opposite, thus scenario 1, 2, 3 are self-contradictive.

    
    Now let's consider the fourth scenario, which implies that both $[b-1]$ and $[b+1]$ are chaiend to some equilibrium state with higher limit payoff for at least one type of agent. Thus there must exist some equilibrium states $[e^+], [e^-] \in S_n$ where $[b-1] \Rightarrow [e^-]$ and $[b+1] \Rightarrow [e^+]$ such that for some $x^-,x^+ \in \{0,1\}$, the following is true
    
    \[u(x^-,[e^-]) > u(1,[b]) \text{, and } \] 
    
    \[u(x^+,[e^+]) > u(0,[b]).\]
    
    It must be true that $e^- < b$ and $e^+ > b$. Because suppose $e^- \ge [b]$ then $[b] \Rightarrow [b]$ by Lemma \ref{monotone lemma}. This implies that $[b] \rightarrow [b]$  and the fourth scenario is equilvalent to one of the first three, which we have proved as contradictive. The same logic can be applied to $e^+$. 
    
    Now let's consider agent $a_i$ who switched from defective at $[b]$ to cooperative at $[b-1]$, there will always exist such agent because $[b] \rightarrow [b-1]$. Since $[b-1] \Rightarrow [e^-]$, such agent will eventually end up in state $[e^-]$ and accept its equilibrium payoff at $[e^-]$. Since we haven't proved that every state lead to exactly one state, there can exist multiple equilibrium states that satisfy the conditions for $[e^-]$, let the set $E^-$ contain all states that can be $e^-$. Since we have established an ordering of states, let $[e^*]$ be the highest-ordered element of $E^-$. 
    
    Suppose $a_i$ is defective at $[e^*]$, then $[e^*] < b-1$. By Lemma 3.1, we know that there exists a monotone chain of connections from $[b]$ to $[e^*]$, and somewhere along this chain, there must be a state $[c] > [e^*]$ where $a_i$ reverts from cooperative to defective. We only allow one state switch at a time, so $[c]$ must lead to $[c+1]$. Since $[c+1]$ leads to $[c]$ as well, it must be true that $[c] \rightarrow [c]$. This dynamic is shown in the diagram below. Observe that $c \in E^-$ and has higher order than $e^-$, which gives us a contradiction. So $a_i$ must be cooperative at $e^-$.
    
    
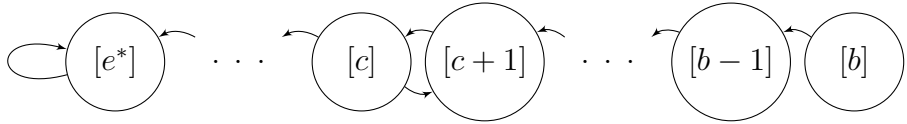
\vspace{0.5cm}
    
    \begin{figure}[h]\

\centering
    \begin{tikzpicture}
  [scale=.54,auto=left,every node/.style={draw, circle, thin, minimum size=1.3cm}]
  
  \tikzset{edge/.style = {->,> = latex'}}
  \node[bag] (n1) at (19,0) {$[b]$};
  \node[bag] (n2) at (16,0)  {$[b-1]$};
  \node[bag, draw = none] (n3) at (13,0) {. . .};
  \node[bag] (n4) at (10,0) {$[c+1]$};
  \node[bag] (n5) at (7,0) {$[c]$};
  \node[bag, draw = none] (n6) at (4,0) {. . .};
  \node[bag] (n7) at (1,0) {$[e^*]$};
  
  
  \draw[edge] (n1)  to[bend right] (n2);
  \draw[edge] (n2)  to[bend right] (n3);
  \draw[edge] (n3)  to[bend right] (n4);
  \draw[edge] (n4)  to[bend right] (n5);
  \draw[edge] (n5)  to[bend right] (n6);
  \draw[edge] (n5)  to[bend right] (n4);
  \draw[edge] (n6)  to[bend right] (n7);
  \draw[edge] (n7)  to[loop left] (n7);

\end{tikzpicture}
\end{figure}

\vspace{0.5cm}
    
    Plugging $x^- = 0$ into the previous expression, we obtain
    
    \[u(0,[e^-]) > u(1,[b]).\]
    
    Now consider the state $[b+1]$. By the monotone decreasing property of the game's utility function, we obtain \[u(1,[b+1]) < u(1,[b]) < u(0,[e^-]).\]
    
    Drawing from the locally non-cooperative property of the game, we can expand this inequality to 
    
    \[u(0,[b+1]) < u(1,[b+1]) < u(1,[b]) < u(0,[e^-]).\]
    
    Suppose there is a defective agent $a_i$ at $[b+1]$, its payoff will not increase if the game reach any higher ordered state no matter whether it's defective or cooperative at that state. However, the cooperative payoff at $[e^-]$ is higher than its current defective payoff. Since $[b] \Rightarrow [e^-]$, and the agent who induced $[b+1]$ into $[b]$ must be cooperative at $[e^-]$, Switching from defective to cooperative is the only feasible strategy for $a_i$ to increase its payoff, therefore $[b+1] \rightarrow [b]$, which means $[b] \rightarrow [b]$, and the fourth scenario is analogous to the first which is contradictive. 
    
    Suppose there is no defective agent at $[b+1]$, then $[b+1] = [0]$, and $[b] = [-1]$, which is not a valid state of $G$.


    
\end{proof}


\begin{corollary}
    If $G^*$ has monotonously decreasing utility function, and is locally  non-cooperative, then the sequence of state switches for $G^*$ contains exactly one cycle, and such cycle can only contain one state.
\end{corollary}

\begin{proof}
    The game $G^*$ is played on a infinite horizon with finite number of unique states. Thus $G^*$ is guaranteed to reach some state twice which creates at least one cycle. Then because every state leads to exactly one state, once the game reaches a cycle of state switches it can not reach any other state outside of such cycle. Therefore $G^*$ contains exactly one cycle.
    
    To show that a cycle in $G^*$ can only contain one state, we first show it can contain at most two states then show a 2-state cycle is impossible. In order to create a cycle of length greater than 2, some state $[b]$ must lead to another state $[s]$ where $s \le b-2$ or $s \ge b+2$. This is unfeasible since $[b]$ can only lead to $[b]$, $[b+1]$, or $[b-1]$. Now imagine we have a 2-state cycle, which means $[b] \rightarrow [b+1]$ and $[b+1] \rightarrow [b]$ for some $b \in \{0, 1,..., n-1\}$. In this scenario, $[b] \rightarrow [b+1]$ implies one player at $[b]$ is capable of inducing a better equilibrium (limit of the means) payoff by transitioning to $[b+1]$, implying $[b]$ and $[b+1]$ lead to different equilibriums payoff vectors. However, the equilibrium payoff vector that $[b]$ and $[b+1]$ are chained to are the same, namely, the payoff of cycling between $[b]$ and $[b+1]$. This creates a contradiction. 
\end{proof}

Notice that an immediate corollary from the above lemma is that every state is chained to exactly one equilibrium state, which means the payoff per round for a state $s$ is asymptotically the payoff of at the equilibrium state that it is chained to. 

\begin{lemma}
     Suppose $G^*$ has monotonously decreasing utility function, and is locally  non-cooperative. If $[b] \rightarrow [b+1]$ for some $b \in \{0,1,...,n-1\}$, then there exists some equilibrium state $[e]$ such that $e \ge b+1$ and $[d] \Rightarrow [e]$ for all $d \in \{b+1, b+2, ...,e-1\}$.  
\end{lemma}

\begin{proof}
    Given $[b] \rightarrow [b+1]$, we know that $[b] \Rightarrow [e]$ for some $e \ge b+1$. By Theorem 3.7, we know each state $[d]$ satisfying $b \le d < e$ can only lead to one state. If for any $d \in \{b+1,b+2,...,c-1\}$, it isn't true that $[d] \rightarrow [d+1]$, then $[d] \rightarrow [d-1]$, which implies $[b] \not \Rightarrow [c]$. Therefore it must be true that $[d] \rightarrow [d+1]$, which means $[d] \Rightarrow [e]$. 
\end{proof}

\begin{theorem}
    Suppose $G^*$ has monotonously decreasing utility function, and is locally  non-cooperative. If state $[b] \Rightarrow [c]$ whereas $c > b$, then there must exist some equilibrium state $[e]$ such that $e \ge c$ and $u(1,[e]) > u(0,[b])$. Similarly, if $[b] \Rightarrow [c]$ and $c < b$, then $e \le c$ and $u(0,[e]) > u(1,[b])$.
\end{theorem}

\begin{proof}
    The two cases in the claim are symmetrical, so it suffices to show one of the cases is true. If $[b] \Rightarrow [c]$ and $c > b$ then we know that $[b] \rightarrow [b+1]$. This implies a cooperative agent $a_i$ at $[b]$ is able to let the game reach a better equilibrium state $[e]$ for itself by switching to defecting. If we show $a_i$ must be defecting at $[e]$, it immediately follows that $u(1,[e]) > u(0,[b])$. 
    
    We know that $[b+1] \Rightarrow [e]$, thus the state transition sequence going from $[b+1]$ to $[d]$ must only contain transitions in the form of $[h] \rightarrow [h+1]$ for $h\in \{b+1,b+2,...,d-1\}$.  Since only one agent is allowed a strategy switch between any state transition, this implies no player is switching from defecting to cooperating in any of these state transitions. Given $a_i$ is defecting at $[b+1]$, it must be defecting at $[e]$ as well. 
\end{proof}

\begin{theorem}
    Suppose $G^*$ has monotonously decreasing utility function, and is locally  non-cooperative, then for any $b \in \{0,1,...,n\}$, 
    
    \begin{itemize}
        \item $[b] \rightarrow [b-1]$ implies there exists some state $[e]$ such that $e \le b-1$, $[e] \rightarrow [e]$, and for any $d \in \{b, b-1, ..., e+1\}$, $u(0,[e]) > u(1, [d])$;
        \item $[b] \rightarrow [b+1]$ implies there exists some state $[e]$ such that $e \ge b+1$, $[e] \rightarrow [e]$, and for any $d \in \{b, b+1, ..., e-1\}$, $u(1,[e]) > u(0, [d])$.

    \end{itemize}
\end{theorem}

\noindent \textit{Proof Sketch.}    This theorem is a combination of Theorem 3.9 and definition of leading.

\vspace{0.3cm}

In the class of supergames that we have defined, solving the game is equivalent to determine the leading relationship between the states, which can be represented via a directed graph. The states can be represented as vertices and we draw an edge from $s_i$ to $s_j$ if $s_i \rightarrow s_j$, a state is chained to another if there exists a path between the two. Based on the above derivation, we expect the graph to contain disjoint paths each with a cycle on the end vertex.

\begin{lemma}
    If $G^*$ has monotonously decreasing utility function and is locally  non-cooperative, then the state $[n]$ is an equilibrium.
\end{lemma}

\begin{proof}
    The state $[n]$ can only lead to $[n]$ or $[n-1]$. Suppose $[n] \rightarrow [n-1]$, then by Theorem 3.10, there exists a state $[e]$ such that for any state $[d]$ between $[n]$ and $[e]$,  $u(0,[e]) > u(1,[d])$. However, given $G^*$ is locally non-cooperative, $u(1, [e+1]) > u(0,[e])$. Therefore $[n] \not \rightarrow [n-1]$ and $[n] \rightarrow [n]$.
\end{proof}

\begin{corollary}
    A supergame $G^*$ that has monotonously decreasing utility function and is locally non-cooperative will always have at east one symmetric equilibrium.
\end{corollary}

\begin{proof}
    The state $[n]$ is symmetric, and is always an equilibrium in this case by Lemma 3.11.
\end{proof}
    
\begin{lemma}
    If $G^*$ has monotonously decreasing utility function and is locally  non-cooperative, then for some state $[b] \in S_n$ such that $u(1,[n]) > u(0,[b])$, it must be true that $[b] \Rightarrow [n]$.
\end{lemma}

\begin{proof}
    We will first consider whether $[b] \rightarrow [b-1]$. By Theorem 3.11, if $[b] \rightarrow [b-1]$ then there exists some equilibrium state $[e]$ such that $e \le b-1$, and for any $d \in \{b, b-1, ..., e+1\}$, $u(0,[e]) > u(1, [d])$. This condition will not hold if we let $[d] = [e+1]$, as the locally non-cooperative property suggests $u(1,[e+1]) > u(0,[e])$.
    
    Thus $[b]$ either lead to itself or $[b+1]$. If one cooperative agent at $[b]$ can improve its equilibrium payoff by switching to defect then $[b] \not \rightarrow [b]$, and $[b] \rightarrow [b+1]$. From the assumption that $u(1,[n]) > u(0,[b])$, we obtain 
    
    \[u(1,[n-1]) > u(0,[b]),\]
    
    because $u(1,[n-1]) > u(1,[n])$ by the locally non-cooperative property. By repeated applying this inequality to all states between $[b]$ and $[n]$, we can infer that the defective utility at all states higher than $[b]$ is greater than cooperative utility at $[b]$. Since a single sequence of state switch in $G^*$ can only contain 1-cycles, if a cooperative switches to defect at $[b]$, they are guaranteed to be defective at some state higher than $[b]$, which improves their payoff. The equilibrium utility of inducing a state switch from $[b]$ to $[b+1]$ is higher than staying idle, thus $[b] \rightarrow [b+1]$ if $u(1,[n]) > u(0,[b])$. 
    
    At the state $[b+1]$, observe that the condition $u(1,[n]) > u(0,[b+1])$ holds by the locally non-cooperative property. We can apply logic of the previous paragraph again to arrive at $[b+1] \rightarrow [b+2]$. Repeating this process, we obtain 
    
    \[[b] \rightarrow [b+1] \text{, } \ [b+1] \rightarrow [b+2]\text {, ... , }  [n-1] \rightarrow [n],\]
    
    which means $[b] \Rightarrow [n]$.

\end{proof}    

\begin{theorem}
    A supergame $G^*$ that has monotonously decreasing utility function and is locally non-cooperative can have only symmetric equilibria. 
\end{theorem}

\begin{proof}
    We prove this statement by providing an example. Consider an instance of $G^*$ with three players and stage game payoff function defined as below.

\begin{table}[h]
\centering
\begin{tabular}{llcccc}
                                                      &                                    & \multicolumn{4}{c}{\textbf{State}}                                                                                        \\ \cline{2-6} 
\multicolumn{1}{l|}{}                                 & \multicolumn{1}{l|}{}              & \multicolumn{1}{c|}{{[}0{]}} & \multicolumn{1}{c|}{{[}1{]}} & \multicolumn{1}{c|}{{[}2{]}} & \multicolumn{1}{c|}{{[}3{]}} \\ \cline{2-6} 
\multicolumn{1}{l|}{\multirow{2}{*}{\textbf{Action}}} & \multicolumn{1}{l|}{0 (Cooperate)} & \multicolumn{1}{c|}{6}       & \multicolumn{1}{c|}{3}       & \multicolumn{1}{c|}{1}       & \multicolumn{1}{c|}{-}       \\ \cline{2-6} 
\multicolumn{1}{l|}{}                                 & \multicolumn{1}{l|}{1 (Defect)}    & \multicolumn{1}{c|}{-}       & \multicolumn{1}{c|}{7}       & \multicolumn{1}{c|}{5}       & \multicolumn{1}{c|}{4}       \\ \cline{2-6} 
\end{tabular}
\end{table}

Observe that this utility layout is monotonously decreasing and locally non-cooperative. By Lemma 3.11, we can conclude that $[3]$ is an equilibrium in $G^*$. By Lemma 3.12, $[1] \Rightarrow [3]$ and $[2] \Rightarrow [3]$. We can draw the leading relationship with the following directed graph. 

\newpage

\begin{figure}[h]
\centering
    \begin{tikzpicture}
  [scale=.8,auto=left,every node/.style={draw, circle, thin}]
  
  \tikzset{edge/.style = {->,> = latex'}}
  \node[bag] (n1) at (8,12) {[0]\\ (6, 6, 6)};
  \node[bag] (n2) at (12,10)  {[1]\\ (3, 3, 7)};
  \node[bag] (n3) at (8,8) {[2]\\ (1, 5, 5)};
  \node[bag] (n4) at (4,10) {[3]\\ (4, 4, 4)};
  
  \draw[edge] (n2)  to[bend left] (n3);
  \draw[edge] (n3)  to[bend left] (n4);
  \draw[edge] (n4)  to[loop left] (n4);
\end{tikzpicture}
\caption{The state diagram representing the leading relationships. Each state has its label on top and the corresponding payoff of agents in the order of cooperative to defective from left to right at the bottom. }
\end{figure}
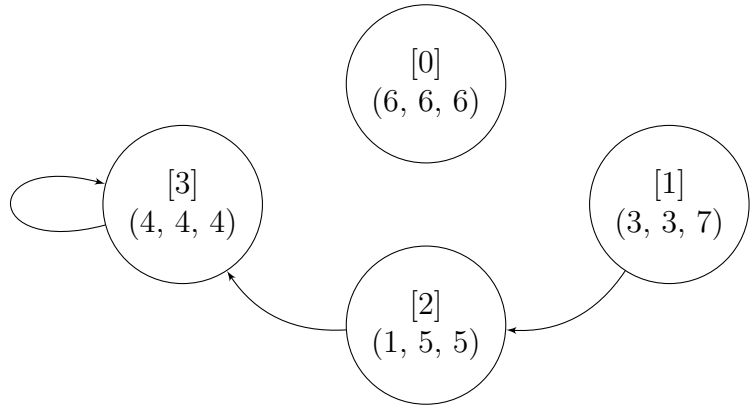

The state transition diagram in Figure 1 indicates that $[1],[2],[3] \Rightarrow [3]$. This means that if the game reaches any one of those states it will be in $[3]$ for the infinitely many rounds after. Now let's consider what $[0]$ leads to, if it leads to $[1]$ then every agent has limit of means payoff of $4$. If $[0]$ leads to itself then the agents can have payoff of $6$. Thus $[0] \rightarrow [0]$ and we have completed the state leading diagram.

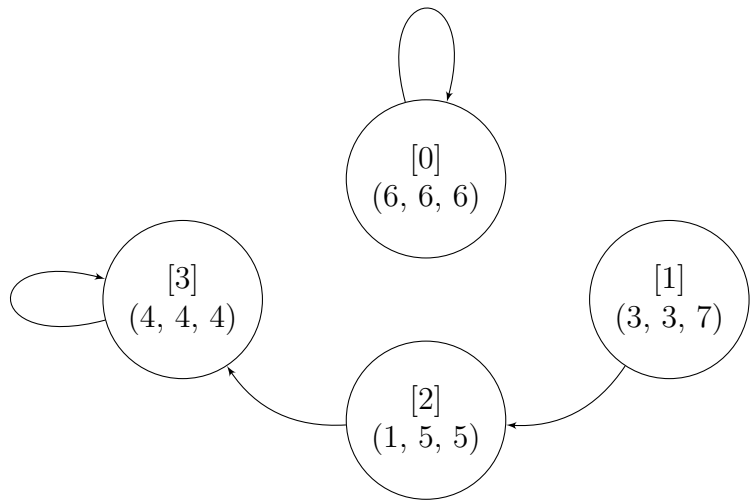
\begin{figure}[h]
\centering
    \begin{tikzpicture}
  [scale=.8,auto=left,every node/.style={draw, circle, thin}]
  
  \tikzset{edge/.style = {->,> = latex'}}
  \node[bag] (n1) at (8,12) {[0]\\ (6, 6, 6)};
  \node[bag] (n2) at (12,10)  {[1]\\ (3, 3, 7)};
  \node[bag] (n3) at (8,8) {[2]\\ (1, 5, 5)};
  \node[bag] (n4) at (4,10) {[3]\\ (4, 4, 4)};
  
  \draw[edge] (n2)  to[bend left] (n3);
  \draw[edge] (n3)  to[bend left] (n4);
  \draw[edge] (n1)  to[loop above] (n1);
  \draw[edge] (n4)  to[loop left] (n4);
\end{tikzpicture}
\caption{The completed diagram. }
\end{figure}

The equilibria in this game are $[0]$ and $[3]$, which are both symmetric. 

\end{proof}

\begin{theorem}
    A supergame $G^*$ that has monotonously decreasing utility function and is locally non-cooperative can have both symmetric and asymmetric equilibria. 
\end{theorem}

\begin{proof}
    Consider a three player, locally non-cooperative instance of $G^*$ that has monotonously decreasing payoff function defined as below.


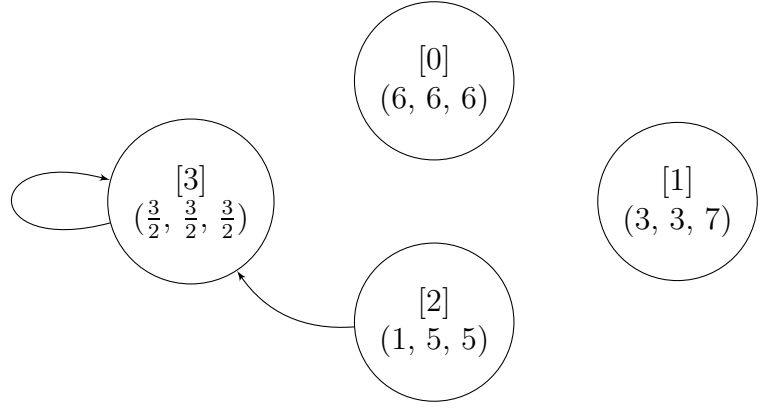
\begin{table}[h]
\centering
\begin{tabular}{llcccc}
                                                      &                                    & \multicolumn{4}{c}{\textbf{State}}                                                                                              \\ \cline{2-6} 
\multicolumn{1}{l|}{}                                 & \multicolumn{1}{l|}{}              & \multicolumn{1}{c|}{{[}0{]}} & \multicolumn{1}{c|}{{[}1{]}} & \multicolumn{1}{c|}{{[}2{]}} & \multicolumn{1}{c|}{{[}3{]}}       \\ \cline{2-6} 
\multicolumn{1}{l|}{\multirow{2}{*}{\textbf{Action}}} & \multicolumn{1}{l|}{0 (Cooperate)} & \multicolumn{1}{c|}{6}       & \multicolumn{1}{c|}{3}       & \multicolumn{1}{c|}{1}       & \multicolumn{1}{c|}{-}             \\ \cline{2-6} 
\multicolumn{1}{l|}{}                                 & \multicolumn{1}{l|}{1 (Defect)}    & \multicolumn{1}{c|}{-}       & \multicolumn{1}{c|}{7}       & \multicolumn{1}{c|}{5}       & \multicolumn{1}{c|}{1.5} \\ \cline{2-6} 
\end{tabular}
\end{table}

By Lemma 3.11, we know $[3] \rightarrow [3]$ and by lemma 3.12, we know $[2] \rightarrow [3]$. 

\vspace{0.2cm}

\begin{figure}[h]
\centering
    \begin{tikzpicture}
  [scale=.8,auto=left,every node/.style={draw, circle, thin}]
  
  \tikzset{edge/.style = {->,> = latex'}}
  \node[bag] (n1) at (8,12) {[0]\\ (6, 6, 6)};
  \node[bag] (n2) at (12,10)  {[1]\\ (3, 3, 7)};
  \node[bag] (n3) at (8,8) {[2]\\ (1, 5, 5)};
  \node[bag] (n4) at (4,10) {[3]\\ ($\frac{3}{2}$, $\frac{3}{2}$, $\frac{3}{2}$)};

  \draw[edge] (n3)  to[bend left] (n4);
  \draw[edge] (n4)  to[loop left] (n4);
\end{tikzpicture}
\end{figure}
    
Now let's consider whether it would be beneficial for an cooperative or defective agent at $[1]$ to induce a state switch. First, observe that suppose an cooperative agent at $[1]$ switches to defective. Such action will induce the game to reach $[2]$ and then $[3]$ in the infinitely many rounds afterward since $[2] \rightarrow [3]$ and $[3] \rightarrow [3]$. Therefore cooperative agents at $[1]$ will not induce state switch, and $[1] \not \rightarrow [2]$. On the other hand, the defective agents at $[1]$ already achieve the highest single round payoff possible: 7. Because the cooperative agents will not induce state switch at $[1]$, if the defective agent remain idle, they can maintain their current payoff of $7$. Thus $[1] \rightarrow [1]$.

Finally, for a cooperative agent at $[0]$, if they choose to induce state switch from $[0]$ to $[1]$, they can increase their equilibrium payoff from 6 to 7, which means $[0] \rightarrow [1]$. We can now complete the state leading diagram.

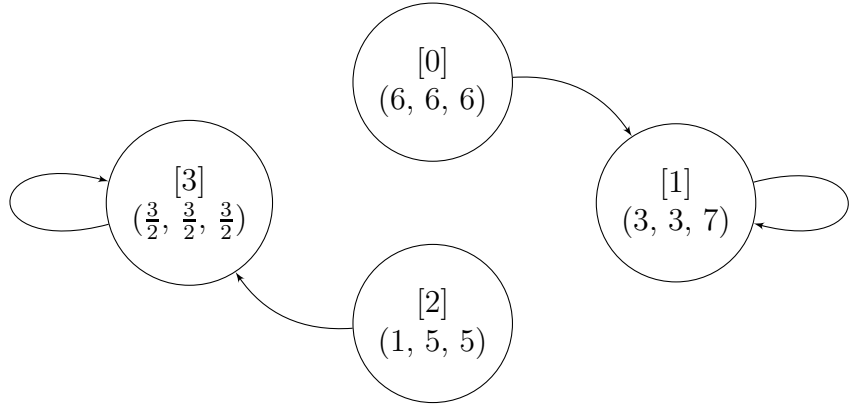
\begin{figure}[h]
\centering
    \begin{tikzpicture}
  [scale=.8,auto=left,every node/.style={draw, circle, thin}]
  
  \tikzset{edge/.style = {->,> = latex'}}
  \node[bag] (n1) at (8,12) {[0]\\ (6, 6, 6)};
  \node[bag] (n2) at (12,10)  {[1]\\ (3, 3, 7)};
  \node[bag] (n3) at (8,8) {[2]\\ (1, 5, 5)};
  \node[bag] (n4) at (4,10) {[3]\\ ($\frac{3}{2}$, $\frac{3}{2}$, $\frac{3}{2}$)};

  \draw[edge] (n3)  to[bend left] (n4);
  \draw[edge] (n1)  to[bend left] (n2);
  \draw[edge] (n2)  to[loop right] (n2);
  \draw[edge] (n4)  to[loop left] (n4);
\end{tikzpicture}
\end{figure}    

The equilibrium states in this example are $[3]$ and $[1]$, which is a mixture of symmetric and asymmetric states. 

\end{proof}

\section{Open Problem}



One obvious limitation of the proposed solution concept is that it only allows one agent to switch strategy at any given round. In this paper, this condition is necessary to proving each state leads to exactly one state (Theorem 3.7) which is critical to every subsequent theorem. However, we suspect the ``one agent at a time" rule might not be a necessary condition to prove Theorem 3.7. Thus, proving Theorem 3.7 with weakened assumption about number of strategy switches per round is likely necessary for generalizing the results in this paper.

\bibliographystyle{plain} 
\bibliography{refs} 

\end{document}